\documentclass[aps,prb,superscriptaddress,amsmath,amssymb,reprint,floatfix]{revtex4-2}
\usepackage[colorlinks=true, urlcolor=blue, linkcolor=blue, citecolor=blue]{hyperref}
\usepackage[dvipsnames]{xcolor}
\usepackage[utf8]{inputenc}
\usepackage{graphicx}
\usepackage{bm}
\usepackage{booktabs}
\usepackage{natbib}
\usepackage{algorithm}
\usepackage{algorithmic}
\usepackage{amsthm}

\graphicspath{{figures/}}

\newtheorem{theorem}{Theorem}
\theoremstyle{definition}
\newtheorem{definition}{Definition}

\newcommand{\s}{\bm{s}}

\begin{document}

\title{Exact autoregressive sampling of planar Ising spin glasses via the Kac--Ward theory}

\author{Jing Liu}
\email{jing.liu@bupt.edu.cn}
\affiliation{School of Physical Science and Technology, Beijing University of Posts and Telecommunications, Beijing 100876, China}
\affiliation{Institute of Theoretical Physics, Chinese Academy of Sciences, Beijing 100190, China}

\author{Tao Chen}
\affiliation{Hefei National Laboratory for Physical Sciences at the Microscale and Department of Modern Physics, University of Science and Technology of China, Hefei 230026, China}
\affiliation{Hefei National Laboratory, University of Science and Technology of China, Hefei 230088, China}

\author{Tianrui Che}
\affiliation{School of Advanced Interdisciplinary Sciences, University of Chinese Academy of Sciences, Beijing 101408, China}

\author{Lei Wang}
\affiliation{Beijing National Laboratory for Condensed Matter Physics and Institute of Physics, Chinese Academy of Sciences, Beijing 100190, China}

\author{Youjin Deng}
\affiliation{Hefei National Laboratory for Physical Sciences at the Microscale and Department of Modern Physics, University of Science and Technology of China, Hefei 230026, China}
\affiliation{Hefei National Laboratory, University of Science and Technology of China, Hefei 230088, China}

\author{Pan Zhang}
\email{panzhang@itp.ac.cn}
\affiliation{Institute of Theoretical Physics, Chinese Academy of Sciences, Beijing 100190, China}
\affiliation{School of Fundamental Physics and Mathematical Sciences, Hangzhou Institute for Advanced Study, UCAS, Hangzhou 310024, China}

\begin{abstract}
Independent samples from the Boltzmann distribution are valuable for estimating spin-glass observables and evaluating approximate sampling methods.
Local Markov chain Monte Carlo methods typically generate temporally correlated samples and can encounter critical slowing down or metastable trapping, while neural autoregressive samplers such as variational autoregressive networks generate independent but approximate samples.
Here we present an exact autoregressive sampling algorithm for planar Ising spin glasses based on the Kac--Ward theory.
Under the chain-rule factorization, sequentially fixing spins induces boundary-localized external fields, which destroy the zero-field structure required for exact evaluation.
By encoding these fields with a planarity-preserving auxiliary spin construction, the conditional partition functions are mapped to an extended zero-field Ising model and exactly evaluated using the Kac--Ward determinant formula.
The method generates strictly independent and identically distributed samples with exact normalized likelihoods at a computational cost of $\mathcal{O}(N^{5/2})$ for $N$ spins, thereby providing an exact baseline for benchmarking neural autoregressive samplers.

\end{abstract}

\maketitle

\section{Introduction}\label{sec:introduction}
The efficient generation of equilibrium configurations from the Boltzmann distribution is a fundamental problem in statistical physics.
For decades, Markov chain Monte Carlo (MCMC) methods~\cite{newman1999monte} such as the Metropolis--Hastings algorithm~\cite{metropolis1953equation} have served as the standard computational approach.
However, standard local MCMC approaches face severe limitations when applied to systems exhibiting phase transitions or complex energy landscapes.
Near the critical region, the divergence of the correlation length leads to critical slowing down, where the integrated autocorrelation time diverges as a power law in the system size.
Furthermore, in disordered systems such as spin glasses~\cite{sfedwards1975theory} or geometrically frustrated antiferromagnets, the energy landscape is partitioned by high barriers.
Consequently, local update algorithms often become trapped in metastable states, failing to traverse the configuration space ergodically.
While advanced cluster algorithms (e.g., Swendsen--Wang~\cite{swendsen1987nonuniversal}, Wolff~\cite{wolff1989collective}) successfully mitigate critical slowing down in ferromagnetic models, they generally lose effectiveness in the presence of frustration.
Currently, parallel tempering (PT)~\cite{hukushima1996exchange,e.marinari1992simulated} still stands as the state-of-the-art method for simulating glassy systems.
However, it requires the laborious tuning of temperature schedules and remains computationally expensive.
Thus, developing methods that can accurately sample equilibrium configurations for spin glasses remains an active area of research.

Recently, variational autoregressive networks (VANs)~\cite{wu2019solving,bialas2022hierarchical,liu2025efficient,li2025deep,delbono2026demonstrating} and tensor-network-based methods~\cite{frias-perez2023collective,chen2025tensor,chen2025tnmcmc} have emerged as promising alternatives for approximately sampling from the Boltzmann distribution.
Unlike MCMC, these methods generate independent and identically distributed (i.i.d.) samples free of autocorrelation.
When combined with importance sampling or Metropolis--Hastings acceptance steps~\cite{nicoli2020asymptotically,mcnaughton2020boosting}, the resulting samples are asymptotically unbiased.
However, the practical efficiency of such steps degrades when the variational ansatz is inaccurate, a regime that is difficult to detect without exact reference samples.
In addition to evaluating these variational methods, the reliable estimation of key spin-glass observables---such as the overlap distribution, spin-glass susceptibility, and the Binder cumulant---requires evaluating multi-replica correlations and higher-order moments, which critically depends on obtaining well-equilibrated and uncorrelated samples from the Boltzmann distribution.

For the planar Ising model, the graph-theoretic approach to the exact computation of the partition function was pioneered by the seminal work of Kac and Ward~\cite{kac1952combinatorial}.
Alternatively, Kasteleyn and Fisher independently formulated exact solutions by mapping the zero-field Ising model to a perfect matching problem on decorated lattices, where the partition function is evaluated via the Pfaffian of an antisymmetric matrix~\cite{kasteleyn1963dimer,fisher1966dimer}.
Building upon this Pfaffian approach, an exact sampling algorithm was elegantly achieved in Ref.~\cite{thomas2009exact} that maps the two-dimensional (2D) spin glass to perfect matching on a Fisher graph.
By employing Wilson's algorithm~\cite{wilson1997determinant} alongside a nested dissection strategy, their method recursively samples spins along macroscopic geometric separators.
However, this divide-and-conquer approach is intrinsically hierarchical rather than sequential: it does not explicitly provide the spin-by-spin conditional probabilities $P(s_i \vert \s_{<i})$, which are strictly required to evaluate the chain-rule factorization under a canonical autoregressive ordering.

The theoretical possibility of constructing such an exact autoregressive sampler, i.e., by computing the partition functions conditioned on previously assigned spins, has been noted in several works~\cite{thomas2009exact, krauth2006statistical}.
Yet, its realization has remained elusive.
The fundamental bottleneck is that autoregressively fixing a subset of spins naturally induces effective fields on the boundary of the unsampled region.
Because the topological mappings underlying both the Kac--Ward and Kasteleyn--Pfaffian methods are strictly valid only in the absence of external fields, these induced boundary fields destroy the exact loop formulations, thus precluding the exact evaluation of autoregressive conditional probabilities.
As a result, evaluating the fidelity of VANs still relies heavily on computationally expensive PT for baseline comparisons~\cite{delbono2025nearestneighbors,delbono2026demonstrating}.

To bridge this gap, we present an exact autoregressive sampling algorithm for planar Ising models.
The method decomposes the Boltzmann distribution into exact conditional probabilities via the chain rule~\cite{biazzo2023autoregressive,biazzo2024sparse} and evaluates them by applying the Kac--Ward determinant formula~\cite{kac1952combinatorial} to an extended planar graph.
Crucially, we adopt a planarity-preserving auxiliary spin construction that was recently introduced for the exact maximum likelihood decoding of quantum error correction codes~\cite{cao2025exact,feng2025planar} to encode the effective boundary fields induced during the sequential generation (see Fig.~\ref{fig:illustration}).
The algorithm generates strictly i.i.d. samples with a computational complexity of $\mathcal{O}(N^{5/2})$ for a system of $N$ spins.
This provides the exact normalized likelihoods necessary to rigorously benchmark modern neural-network and tensor-network methods.
Throughout this work, ``exact'' refers to sampling the Boltzmann distribution without systematic error in exact arithmetic, while finite-precision implementations remain subject to roundoff errors.
A reference Python implementation accompanies this work~\cite{pykacward}.
This paper is organized as follows.
In Sec.~\ref{sec:methods} we describe the autoregressive factorization, the auxiliary spin construction, and an efficient scheme for evaluating the determinant ratio via the matrix determinant lemma.
In Sec.~\ref{sec:results} we benchmark the algorithm on representative planar Ising systems.
In Sec.~\ref{sec:discussion} we conclude with a discussion.

\section{Methods}\label{sec:methods}
\subsection{Autoregressive factorization of the Boltzmann distribution}\label{subsec:autoregressive}
Without loss of generality, we consider an Ising model defined on an open square lattice of size $N = L \times L$.
The spin configuration is denoted by $\s = \{s_1,\cdots,s_N\}$, where each spin takes values $s_i \in \{-1, +1\}$.
With the nearest-neighbor interactions $J_{ij}$ and inverse temperature $\beta = 1/T$, the Hamiltonian is given by $E(\s) = -\sum_{\langle i,j \rangle} J_{ij} s_i s_j$, where $\langle i,j \rangle$ denotes the set of nearest-neighbor pairs.
The associated equilibrium Boltzmann distribution reads
\begin{equation}
    P(\s) = \frac{\mathrm{e}^{-\beta E(\s)}}{Z},
\end{equation}
with the partition function $Z = \sum_{\s} \mathrm{e}^{-\beta E(\s)}$.
The chain rule factorizes the Boltzmann distribution into an exact product of conditional probabilities,
\begin{equation}
    P(\s) = \prod_{i=1}^N P_i(s_i \vert \s_{<i}),
\end{equation}
where $\s_{<i} = \{s_1,\cdots,s_{i-1}\}$ and $\s_{>i} = \{s_{i+1}, \cdots, s_N\}$ denote the sets of spins preceding and succeeding index $i$ under a fixed autoregressive ordering.
Once all the conditional probabilities are available, one can generate i.i.d. configurations from $P(\s)$ by ancestral sampling~\cite{bishop2006pattern}, i.e., by sequentially drawing each spin from its conditional probability.
Following the formulation derived in Refs.~\cite{biazzo2023autoregressive,biazzo2024sparse}, the $i$-th conditional probability is explicitly given by
\begin{equation}\label{eq:conditional-1}
P_i(s_i \vert \s_{<i}) = \frac{\sum_{\s_{>i}}\mathrm{e}^{-\beta E(\s)}}{\sum_{\s_{>i-1}}\mathrm{e}^{-\beta E(\s)}} = \frac{f(s_i,\s_{<i})}{\sum_{s_i}f(s_i,\s_{<i})},
\end{equation}
where $f(s_i,\s_{<i}) = \sum_{\s_{>i}}\mathrm{e}^{-\beta E(\s)}$ represents the conditional partition function obtained by fixing $\s_{<i}$ and $s_i$ while summing over the remaining spins.
By substituting the specific form of the Ising Hamiltonian into Eq.~\eqref{eq:conditional-1}, one obtains the exact expression for the autoregressive conditional probability in terms of the sigmoid function $\sigma(x) = (1 + \mathrm{e}^{-x})^{-1}$:
\begin{equation}\label{eq:twobo}
P_i(s_i=+1\vert\s_{<i}) = \sigma\left( 2 \beta \sum_{j,j<i}J_{ji}s_j + \ln \frac{Q^+_i}{Q^-_i}\right).
\end{equation}
Here the first term corresponds to the direct interactions between spin $i$ and the previously sampled spins, while the second term captures the influence of the unsampled spins through
\begin{multline}
Q^{\pm}_i = \sum_{\s_{>i}} \exp \beta \bigg[\sum_{\langle m,n\rangle,m,n>i} J_{mn} s_m s_n \\ + \sum_{m, m>i} s_m \big(\pm J_{im} + \sum_{j, j<i} J_{jm} s_j\big)\bigg].
\end{multline}

The physical interpretation of Eq.~\eqref{eq:twobo} becomes transparent by introducing~\cite{biazzo2023autoregressive,biazzo2024sparse}
\begin{equation}
    \xi_{im} = \sum_{j<i} J_{jm} s_j
\end{equation}
for $m > i$, which encode the influence of the frozen spins $\s_{<i}$ on the remaining unsampled site $m$ when sampling the $i$-th spin.
Due to the sparse and local nature of the Ising interactions, most $\xi_{im}$ vanish identically, and each remaining nonzero $\xi_{im}$ involves only a fixed number of spins, see the dashed red lines in Fig.~\ref{fig:illustration}(a).
Although $Q_i^{\pm}$ formally involve a summation over all configurations of the unsampled spins, they depend on $\s_{<i}$ only through these nonzero $\{\xi_{im}\}_{m>i}$.
This observation underlies the TwoBo autoregressive architecture~\cite{biazzo2024sparse}, where $Q_i^{\pm}$ are approximated by a neural network, as well as the physics-inspired sparse attention design in Ref.~\cite{zhong2026scalable}.
In addition, by exploiting the tensor-network formulation of the Ising model, $Q_i^{\pm}$ can be evaluated efficiently through approximate tensor-network contraction, as done in the tensor-network Monte Carlo methods of Refs.~\cite{frias-perez2023collective,chen2025tensor}.
However, all of these approaches evaluate $Q_i^{\pm}$ only approximately.
In what follows, we show how to evaluate $Q_i^{\pm}$ exactly.

\begin{figure*}[!ht]
    \centering
    \includegraphics[width=0.9\linewidth]{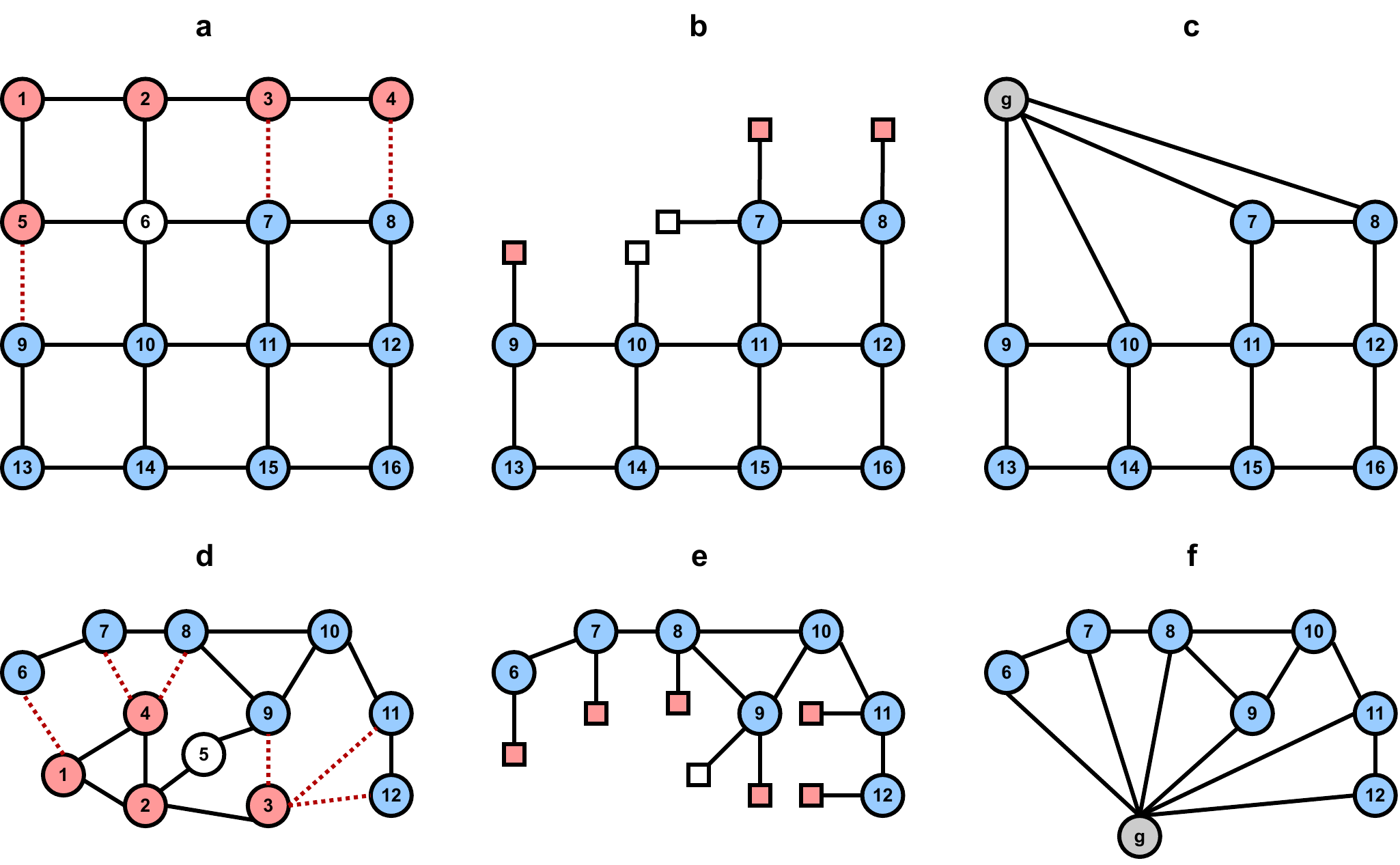}
    \caption{Schematic illustration of the exact autoregressive sampling algorithm via the Kac--Ward formula.
    (a) Snapshot of the sampling process on a $4 \times 4$ square lattice using a raster-scan order.
    The spins are partitioned into three sets: the frozen spins $\s_{<i}$ (red), the current spin $s_i$ being sampled (white), and the unsampled spins $\s_{>i}$ (blue).
    The interactions between the frozen spins and the unsampled spins induce effective local fields on the boundary sites (dashed red lines).
    (b) The subgraph described by Eq.~\eqref{eq:effective-ham}.
    (c) The auxiliary spin construction.
    To evaluate the partition functions $Q_i^{\pm}$ in (b) using the Kac--Ward formula, the effective local fields are replaced by pairwise couplings to an auxiliary spin $s_g$ (gray).
    The auxiliary spin is positioned in the region corresponding to the sampled spins, connecting to the boundary sites without introducing edge crossings, thus preserving the planarity of the extended interaction graph.
    (d)--(f) The same process for a general planar graph.}
    \label{fig:illustration}
\end{figure*}

\subsection{Kac--Ward formula for the planar Ising model}\label{subsec:kacward}
To evaluate $Q_i^{\pm}$ exactly, we exploit the planarity of the subgraph using the Kac--Ward determinant formula~\cite{kac1952combinatorial,johnson2016learning}, which we now recall.
Consider a planar graph $G(V,E)$ embedded in the plane with straight-line edges.
Let $\phi_{ijk} \in [-\pi,+\pi]$ denote the clockwise rotation between the directed edges $(i,j)$ and $(j,k)$.
One defines two matrices $B, D \in \mathbb{C}^{2|E|\times 2|E|}$ by
\begin{align}
B_{ij,jk,k\neq i} &= \exp(\sqrt{-1} \phi_{ijk}/2), \\
    D_{ij,ij} &= \tanh(K_{ij}),
\end{align}
where $K_{ij} = \beta J_{ij}$ is the dimensionless coupling, $B$ is the so-called non-backtracking matrix, and $D$ is diagonal.
The partition function of the zero-field Ising model on $G$ is then given by the Kac--Ward determinant formula:
\begin{multline}\label{eq:kacward}
    \ln Z = N \ln 2 + \sum_{\langle i,j\rangle}\ln\cosh K_{ij} \\ + \frac{1}{2} \ln \det(I - BD).
\end{multline}

\subsection{Evaluation of the conditional probability}\label{subsec:eval-cond}
The quantities $Q_i^{\pm}$ correspond to partition functions of the effective Ising model defined on the subgraph, with Hamiltonian
\begin{equation}\label{eq:effective-ham}
E_{\text{eff}}^{\pm}(\s_{>i}) = -\sum_{\langle m,n \rangle} J_{mn} s_m s_n - \sum_{m} h_m^{\pm} s_m,
\end{equation}
where the effective local fields are given by $h_m^{\pm} = \pm J_{im} + \sum_{j < i} J_{jm} s_j = \pm J_{im} + \xi_{im}$.
In general, the presence of arbitrary external fields renders the exact evaluation of the partition function NP-hard for planar Ising models~\cite{fbarahona1982computational}, as the fields destroy the mapping to non-intersecting loop configurations underlying the Kac--Ward construction~\cite{kac1952combinatorial}.

However, the present setting admits an exact solution due to the specific planar geometry.
For example, on a nearest-neighbor square lattice under the raster-scan ordering, the effective local fields $h_m^{\pm}$ are nonzero only on sites $m$ that are adjacent to the previously sampled region, as we have discussed above.
These sites form a one-dimensional boundary separating the frozen spins $\s_{< i}$ from the remaining subgraph $\s_{>i}$.
As a result, the external fields act exclusively on boundary spins and do not interact with the bulk of the unsampled region.
This observation allows us to recover a purely quadratic Hamiltonian while preserving planarity by introducing an auxiliary spin $s_g$~\cite{cao2025exact,feng2025planar}.
Specifically, we replace the effective local fields by pairwise couplings between $s_g$ and the boundary spins, with coupling strengths $J_{gm}^{\pm} = h_m^{\pm}$, see Fig.~\ref{fig:illustration}(b,c).
The resulting extended Hamiltonian reads
\begin{equation}\label{eq:extended-ham}
E_{\text{eff}}^{\pm}(\s_{>i}, s_g) = -\sum_{\langle m,n \rangle} J_{mn} s_m s_n - \sum_{m} J_{gm}^{\pm} s_g s_m,
\end{equation}
which contains no external fields and is therefore amenable to exact treatment.
Let $\widehat{Q}^{\pm}_i$ be the partition function of this enlarged zero-field Ising model, including the auxiliary spin in the summation,
\begin{equation}
    \widehat{Q}^{\pm}_i = \sum_{s_g} \sum_{\s_{>i}} \mathrm{e}^{-\beta E_{\text{eff}}^{\pm}(\s_{>i}, s_g)}.
\end{equation}
Since the Hamiltonian is invariant under the global flip $s_g \to -s_g$ accompanied by $s_m \to -s_m$ for all $m>i$, it follows immediately that $\widehat{Q}^{\pm}_i = 2 Q^{\pm}_i$.
The overall factor of two cancels in the logarithmic ratio entering Eq.~\eqref{eq:twobo}, and thus does not affect the conditional probabilities.

To apply the Kac--Ward formula, the interaction graph of the extended system must remain planar.
This condition is strictly satisfied provided that the autoregressive generation follows a connected search ordering, where the frozen spins always form a connected subgraph.
Under this condition, the auxiliary spin construction is topologically equivalent to progressively contracting the connected frozen region into a single vertex.
Since planar graphs are closed under edge contraction, the extended interaction graph is rigorously guaranteed to remain planar without introducing edge crossings, as illustrated in Fig.~\ref{fig:illustration}(d)--(f).
We provide a rigorous proof of this planarity preservation in Appendix~\ref{appendix:a}.
Consequently, the auxiliary spin construction provides an exact mapping from the original Ising model with boundary-localized external fields to a zero-field Ising model on an enlarged planar graph.
This mapping enables the exact evaluation of $Q_i^{\pm}$ using the Kac--Ward formula, and hence of the conditional probabilities $P_i(s_i \vert \s_{<i})$ at each autoregressive step.
We provide pseudocode describing the autoregressive sampling process for the 2D Ising spin glass in Algorithm~\ref{alg:sampling} as a concrete example.

\begin{algorithm}[H]
\caption{Exact sampling of 2D Ising spin glass}
\label{alg:sampling}
\begin{algorithmic}[1]
\REQUIRE Couplings $\{J_{ij}\}$ on an $N = L \times L$ open lattice, inverse temperature $\beta=1/T$
\ENSURE Exact sample $\s = \{s_1,\ldots,s_N\} \sim P(\s) = \mathrm{e}^{-\beta E(\s)}/Z$
\STATE Initialize empty configuration $\s \gets \varnothing$
\STATE Label each spin in a raster-scan order, see Fig.~\ref{fig:illustration}(a)
\FOR{$i = 1$ to $N$}
    \STATE $H_i \gets 2 \beta \sum_{j<i} J_{ji} s_j$
    \STATE $h_m^{\pm} \gets \pm J_{im} + \sum_{j<i} J_{jm} s_j$ for all $m>i$
    \STATE Construct extended planar graph $G^{\pm}$ by introducing the auxiliary spin $s_g$, and setting $J_{gm}^{\pm} = h_m^{\pm}$ for the boundary sites $m$
    \STATE $\ln \widehat{Q}_i^{\pm} \gets \text{Kac--Ward formula}(G^{\pm}, \beta)$
    \STATE $P_i(s_i=+1 \vert \s_{<i}) \gets \sigma\left(H_i + \ln \widehat{Q}_i^{+} - \ln \widehat{Q}_i^{-}\right)$
    \STATE Sample $s_i \in \{-1,+1\}$ from $P_i(\cdot \vert \s_{<i})$
    \STATE Append $s_i$ to configuration $\s$
\ENDFOR
\end{algorithmic}
\end{algorithm}

\subsection{Efficient evaluation of the determinant ratio}\label{subsec:det-ratio}
At each ancestral sampling step, Eq.~\eqref{eq:twobo} requires evaluating the log-ratio $\ln(\widehat{Q}^+_i/\widehat{Q}^-_i)$.
Since the extended graph is identical for $\widehat{Q}_i^{+}$ and $\widehat{Q}_i^{-}$, the two cases differ only in the auxiliary spin couplings.
Denoting the corresponding diagonal matrices by $D^{+}$ and $D^{-}$, Eq.~\eqref{eq:kacward} gives
\begin{equation}\label{eq:logratio}
    \ln\frac{\widehat{Q}^+_i}{\widehat{Q}^-_i} = \sum_{b} \ln \frac{\cosh K_{gb}^+}{\cosh K_{gb}^-} + \frac{1}{2} \ln \frac{\det(I-BD^{+})}{\det(I-BD^{-})},
\end{equation}
where the summation runs over the boundary sites coupled to the auxiliary spin.
The key observation is that $h_m^+ = h_m^-$ except on the auxiliary edges adjacent to the $i$-th spin itself, so $D^+-D^-$ is nonzero only there.
By setting $A = I - BD^{-}$, we may write $I - BD^{+} = A + UV^\top$, where $U,V\in\mathbb{C}^{2|E|\times r}$ with $r=\mathcal{O}(1)$ (twice the number of unsampled neighbors of the $i$-th spin).
The matrix determinant lemma then yields
\begin{equation}\label{eq:det_ratio}
    \frac{\det(I-BD^{+})}{\det(I-BD^{-})}
= \det(I_r+V^\top A^{-1}U).
\end{equation}
The full determinant ratio thus collapses to an $r\times r$ determinant, requiring only a single LU factorization of $A$ per autoregressive step.
A detailed computational complexity analysis is given in Appendix~\ref{appendix:b}.

\section{Results}\label{sec:results}
We first demonstrate the sampling efficiency of the algorithm on the 2D Edwards--Anderson (EA) spin glass~\cite{sfedwards1975theory} with Hamiltonian $E(\s) = -\sum_{\langle i,j \rangle} J_{ij} s_i s_j$, where the couplings $J_{ij}$ are drawn from a random distribution (e.g., bimodal or Gaussian).
For disordered systems, local MCMC methods suffer from severe dynamic freezing, leading to extremely long relaxation times for $\beta \gtrsim 1.5$~\cite{mcnaughton2020boosting,ciarella2023machinelearningassisted}.
We quantify this behavior by computing the normalized energy autocorrelation function
\begin{equation}
\Gamma(t) = \frac{\langle E(t') E(t'+t) \rangle - \langle E \rangle^2}{\langle E^2 \rangle - \langle E \rangle^2}.
\end{equation}
As illustrated in Fig.~\ref{fig:acf}, the Metropolis algorithm exhibits a slow, stretched exponential decay in $\Gamma(t)$, reflecting a large integrated autocorrelation time~\cite{mcnaughton2020boosting}.
In contrast, the exact autoregressive sampler constructs each configuration independently from the exact conditional distribution, requiring no equilibration and no reference to a prior state.
The resulting autocorrelation function satisfies $\Gamma(t) = \delta_{t,0}$, confirming that the generated configurations are strictly i.i.d.
This property is particularly advantageous for glassy systems where even PT requires careful tuning of temperature schedules and many MC steps to yield uncorrelated samples.

\begin{figure}[!ht]
    \centering
    \includegraphics[width=\linewidth]{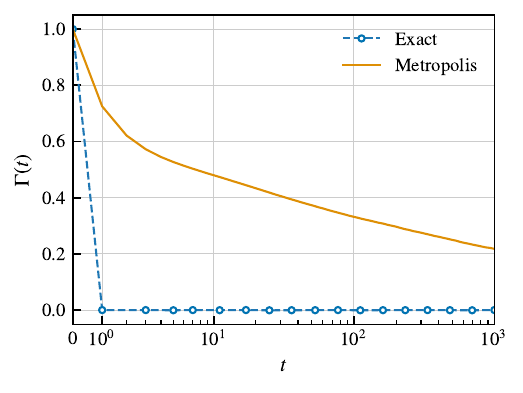}
    \caption{Autocorrelation function $\Gamma(t)$ of the energy for the 2D EA spin glass.
    The simulation was performed on a $16 \times 16$ lattice at temperature $T=0.5$ ($\beta=2$).
    For the exact autoregressive sampler, the lag $t$ is measured in independently generated configurations, whereas for the Metropolis algorithm it is measured in Monte Carlo sweeps, with one sweep corresponding to $N$ attempted single-spin updates.
    The Metropolis algorithm (orange) shows a stretched exponential decay, while the exact autoregressive sampler (blue) yields $\Gamma(t) = \delta_{t,0}$, confirming strictly i.i.d. samples.
    The disorder realizations are taken from Ref.~\cite{delbono2025nearestneighbors}.
    Data are averaged over 100 independent runs.}
    \label{fig:acf}
\end{figure}

Since the algorithm supports arbitrary planar graphs, we next validate it on the triangular ferromagnetic (FM) Ising model~\cite{houtappel1950orderdisorder}, defined by $E(\s) = -J \sum_{\langle i,j \rangle} s_i s_j$ with $J > 0$, where the summation runs over all nearest-neighbor pairs on the triangular grid.
This model exhibits a well-known ferromagnetic phase transition at the critical temperature $T_c = 4J/\ln 3$.
Thus, simulating it near criticality is a computational bottleneck for local MCMC methods due to the critical slowing down.
Figure~\ref{fig:triangular} shows the magnetization per spin $\vert m \vert = \vert N^{-1} \sum_i s_i \vert$ and the specific heat per spin $c_v$ as functions of temperature for system sizes $L = 8$ to $32$, together with the exact specific heat curve for $L=64$, obtained by directly differentiating Eq.~\eqref{eq:kacward} with respect to $\beta$, see Appendix~\ref{appendix:c}.
The specific heat peak sharpens and shifts toward $T_c$ with increasing $L$, consistent with the expected finite-size scaling behavior.
Because our algorithm produces unbiased samples from the exact Boltzmann distribution, these observables are free from the systematic biases that can affect VANs near criticality.

\begin{figure}[!ht]
    \centering
    \includegraphics[width=\linewidth]{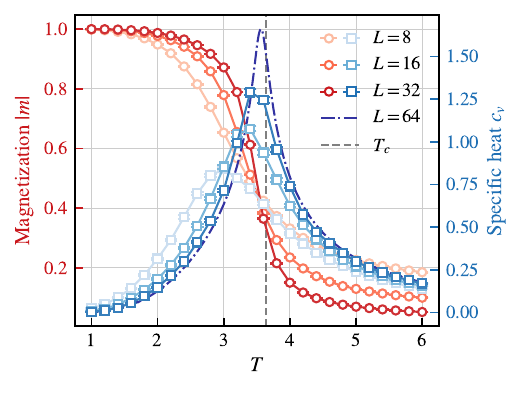}
    \caption{Validation on the triangular FM Ising model.
    The magnetization per spin $|m|$ (red, left axis) and the specific heat per spin $c_v$ (blue, right axis) are plotted as a function of temperature $T$.
    The values are estimated on 10\,000 i.i.d. samples.
    The specific heat curve for $L=64$ is obtained by directly differentiating $\ln Z$ with respect to $\beta$, see Appendix~\ref{appendix:c}.
    The vertical dotted line indicates the theoretical critical temperature $T_c = 4J/\ln 3 \approx 3.641$ ($J=1$).}
    \label{fig:triangular}
\end{figure}

Finally, we demonstrate the exact-likelihood capability of the algorithm on the triangular antiferromagnetic (AFM) Ising model~\cite{wannier1950antiferromagnetism}.
In the thermodynamic limit, the triangular AFM Ising model has an extensively degenerate ground state with residual entropy per spin $s_0\approx0.323$~\cite{wannier1950antiferromagnetism}.
For the finite $L\times L$ open rhombic patches of the triangular lattice studied here, however, the ground state is only twofold degenerate~\cite{zukovic2014entropy}.
At fixed $L$, their entropy per spin therefore approaches $\ln 2/L^2$ as $T\to0$.
Estimating the entropy per spin $s$ requires indirect methods such as thermodynamic integration or Wang--Landau sampling.
Because the algorithm computes the exact normalized probability $P(\s)$ for each generated configuration, the entropy per spin is obtained directly as an unbiased expectation value over $N_s$ samples:
\begin{equation}
s = -\frac{1}{N}\langle \ln P(\s) \rangle \approx -\frac{1}{N N_{s}}\sum_{i=1}^{N_s} \ln P(\s^i).
\end{equation}
Figure~\ref{fig:entropy} presents the temperature dependence of the entropy per spin for the triangular AFM Ising model with system sizes $L=8$ to $32$.
At the plotted temperatures, these estimates agree with the exact entropies obtained from the Kac--Ward partition function (Appendix~\ref{appendix:c}).
The horizontal dashed line marks the bulk residual entropy $s_0$ and serves only as a reference.

\begin{figure}[!ht]
    \centering
    \includegraphics[width=\linewidth]{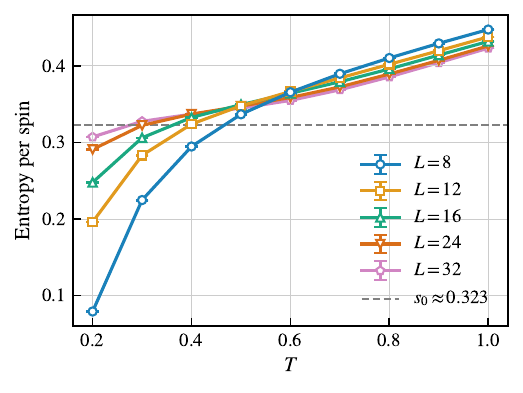}
    \caption{Direct entropy estimation via exact likelihoods.
    Temperature dependence of the entropy per spin for the triangular AFM Ising model with system sizes $L=8$ to $32$.
    The horizontal dashed line marks the bulk residual entropy per spin $s_0\approx0.323$ and only serves as a reference.}
    \label{fig:entropy}
\end{figure}

\section{Discussion}\label{sec:discussion}
In this work, we have presented an exact autoregressive sampling algorithm for arbitrary planar Ising spin glasses.
The method combines the Kac--Ward formula with the auxiliary spin construction, generating strictly unbiased Boltzmann-distributed configurations for any planar coupling arrangement.
The principal computational advantage over conventional Monte Carlo methods is the complete elimination of equilibration and autocorrelation overhead.
Beyond i.i.d. sampling, the method provides direct access to the exact normalized likelihood and conditional probabilities for every generated configuration.
This capability enables unbiased entropy estimation without thermodynamic integration, and the exact likelihoods and unbiased observables provide a rigorous baseline for validating neural autoregressive samplers on the same planar instances used in variational benchmarks~\cite{delbono2025nearestneighbors,delbono2026demonstrating,zhong2026scalable}.

A current limitation of the method is the $\mathcal{O}(N^{5/2})$ computational cost per sample, compared with the $\mathcal{O}(N^{3/2})$ scaling of Ref.~\cite{thomas2009exact}, thus restricting applications to systems of moderate size ($N \approx 10^3$ spins).
This scaling originates from the sequential evaluation of the exact conditional probabilities: each of the $N$ ancestral steps requires a fresh sparse LU factorization of the Kac--Ward matrix, whose computational cost sums to $\mathcal{O}(N^{5/2})$ over a full sweep (see Appendix~\ref{appendix:b}).
In addition, frustrated systems at low temperatures require an increasing number of arithmetic bits to maintain numerical accuracy (see Appendix~\ref{appendix:b}).
The additional cost relative to Ref.~\cite{thomas2009exact} is the price paid for direct access to every per-spin conditional probability, which their divide-and-conquer construction does not explicitly provide.

Looking forward, several directions can extend the utility of the approach.
GPU-accelerated sparse linear algebra packages have the potential to improve the computational efficiency and extend the range of accessible system sizes.
The exact samples also enable high-precision finite-size scaling studies at moderate sizes, where the absence of autocorrelation may compensate for the limited system size.
Beyond benchmarking, the exact conditional probabilities $P_i(s_i \vert \s_{<i})$ open a concrete route to improving these samplers.
Optimizing VANs via the REINFORCE algorithm~\cite{williams1992simple} to minimize the variational free energy suffers from high gradient variance, as the gradient estimator relies on a sparse, global reward signal.
In contrast, our algorithm deterministically computes the conditional probabilities at every generation step, which can serve as per-step teacher signals for on-policy knowledge distillation training~\cite{ross2011reduction,agarwal2024onpolicy,gu2024minillm,lightman2024lets}.

An open-source Python implementation of the algorithm described in this work is available at Ref.~\cite{pykacward}.

\section*{Acknowledgements}
We thank Ying Tang for valuable discussions.
This work is supported by Projects 12405047, 12325501, 12247104, 12447101, 12275263 of the National Natural Science Foundation of China (NSFC), the Strategic Priority Research Program of Chinese Academy of Sciences (under Grant No.XDB1680000), the Quantum Science and Technology-National Science and Technology Major Project (under Grant No.2021ZD0301900), and the Fundamental Research Funds for the Central Universities.

\appendix

\setcounter{figure}{0}
\renewcommand{\thefigure}{S\arabic{figure}}
\renewcommand{\theHfigure}{S\arabic{figure}}

\section{Proof of planarity preservation in exact autoregressive sampling}\label{appendix:a}

In this section, we rigorously prove that the auxiliary spin construction introduced in Sec.~\ref{subsec:eval-cond} preserves the planarity of the interaction graph at every step of the autoregressive sampling process, provided that a valid sampling ordering is employed.
We first formalize the autoregressive sampling step in the language of graph theory.
Let the original Ising model be defined on an undirected, connected planar graph $G = (V, E)$, where $V = \{s_1, \cdots, s_N\}$ represents the set of spin sites and $E$ represents the pairwise interactions.
At step $i$ of the sampling process, the evaluation of the conditional partition functions $Q_i^{\pm}$ requires partitioning the vertex set $V$ into two disjoint subsets: the frozen set $F_i = \{s_1, \cdots, s_i\}$ (which includes the spin $s_i$ currently being sampled, whose two candidate values distinguish $Q_i^+$ from $Q_i^-$), and the remaining unsampled set $R_i = \{s_{i+1}, \cdots, s_N\}$, such that $V = F_i \cup R_i$.

The auxiliary spin construction maps the effective Hamiltonian of the unsampled spins onto an extended graph $G'_i = (V'_i, E'_i)$.
Topologically, this construction replaces the entire frozen region $F_i$ with a single auxiliary vertex $s_g$.
The new vertex set is $V'_i = R_i \cup \{s_g\}$.
The new edge set $E'_i$ retains all original edges internal to $R_i$ and replaces every boundary connection between $F_i$ and $R_i$ with a new edge connected to $s_g$.
To ensure the Kac--Ward formula remains applicable, $G'_i$ must be planar for all $i$.
We establish this through the following theorem.

\begin{definition}[Connected Search Ordering]
An autoregressive sampling ordering on a connected graph is called a connected search ordering if, for every step $i \in \{2, \cdots, N\}$, the subgraph induced by the frozen vertices, denoted as $G[F_i]$, is a connected graph.
\end{definition}

Note that standard graph traversal algorithms, such as breadth-first search or depth-first search starting from an arbitrary initial node, strictly generate connected search orderings on connected graphs.
On the open square lattice, the raster-scan ordering adopted in this work is also a connected search ordering: every prefix of the raster order consists of several complete rows plus an initial segment of the next row, and therefore induces a connected subgraph.

\begin{theorem}
Let $G = (V, E)$ be a connected planar graph.
If the frozen sets $F_i$ are generated via a connected search ordering, then the extended graph $G'_i$ constructed at any step $i$ is guaranteed to be planar.
\end{theorem}

\begin{proof}
The proof relies on the theory of graph minors.
A graph $H$ is a minor of a graph $G$ if $H$ can be formed from $G$ by deleting edges, deleting vertices, and contracting edges.
A fundamental property of planar graphs is that they are minor-closed, i.e., any minor of a planar graph is necessarily planar.
Assume $F_i$ is generated via a connected search ordering, implying that the induced subgraph $G[F_i]$ is connected.
Therefore, $G[F_i]$ contains at least one spanning tree $T$.
Thus we can construct the extended graph $G'_i$ from the original planar graph $G$ using the following sequence of minor-producing operations:
\begin{enumerate}
    \item Edge deletion: Delete all edges in $G[F_i]$ that do not belong to the spanning tree $T$.
    The resulting graph remains planar.
    \item Edge contraction: Sequentially contract every edge belonging to the spanning tree $T$.
    This process will merge the entire subset $F_i$ into a single, unified vertex $s_g$.
    \item Parallel edge merging: If several sites of $F_i$ are adjacent to the same boundary site $s_m \in R_i$, the contraction produces parallel edges between $s_g$ and $s_m$.
    These parallel edges are merged into a single edge with the effective coupling $J_{gm}^{\pm} = h_m^{\pm} = \pm J_{im} + \sum_{j<i}J_{jm}s_j$.
    This ensures that the resulting graph remains a minor of $G$, see the edge between $s_9$ and $s_g$ in Fig.~\ref{fig:illustration}.
\end{enumerate}

Therefore, the graph resulting from these operations is exactly the extended graph $G'_i = (V'_i, E'_i)$ defined by the auxiliary spin construction.
Because $G'_i$ is derived from the planar graph $G$ exclusively through minor-producing operations, $G'_i$ is a minor of $G$.
Since the class of planar graphs is closed under minors, and $G$ is planar, it follows immediately that $G'_i$ must be planar.
This holds for any step $i$.
\end{proof}

\section{Computational complexity and numerical accuracy}\label{appendix:b}
As shown in Sec.~\ref{subsec:det-ratio}, each autoregressive step requires one sparse LU factorization of the matrix $A = I - BD^{-}$ to evaluate the determinant ratio via the matrix determinant lemma.
Summing the per-step factorization cost over all $N$ steps yields the total sampling complexity.
Let $M = N - i + 1$ denote the number of remaining spins at step~$i$, so that the dimension of $A$ scales as $\mathcal{O}(M)$.
Each sparse LU factorization costs $\mathcal{O}(M^{3/2})$ for a planar graph.
The total cost is dominated by the early steps where $M \sim N$, giving $\mathcal{O}(N^{5/2})$ overall at fixed arithmetic precision, as shown in Fig.~\ref{fig:appendix}(a).

One might hope to reduce the per-step cost by replacing the direct solver with iterative Krylov methods, thereby exploiting the $\mathcal{O}(N)$ sparsity of $A$.
However, in practice we find that $A = I - BD^{-}$ is ill-conditioned for frustrated systems at low temperatures.
As shown in Fig.~\ref{fig:appendix}(b), its condition number remains between $1$ and $10^{2}$ for non-frustrated models (the Ising model and the triangular FM Ising model).
The condition number diverges exponentially under quenched disorder (the EA spin glass) or geometric frustration (the triangular AFM Ising model).
Maintaining a prescribed numerical accuracy in this regime therefore requires an increasing number of arithmetic bits as the temperature is lowered.
A similar issue is encountered in Ref.~\cite{thomas2009exact}.

To quantify the finite-precision error in practice, we perform a direct numerical stability test, shown in Fig.~\ref{fig:appendix}(c).
For each sampled configuration, we compute $\ln \widehat{Z}(\s) = -\beta E(\s) - \ln \widehat{P}(\s)$, which should equal the exact value $\ln Z$.
The reference value $\ln Z_{\mathrm{ref}}$ is computed from the Kac--Ward determinant formula in 200-bit floating-point arithmetic.
Repeating the calculation with 100 bits gives the same results, confirming that the reference values are converged.
The deviation of $\ln \widehat{Z}(\s)$ from $\ln Z_{\mathrm{ref}}$ therefore measures the accumulated roundoff error over a complete autoregressive trajectory.
We test 1000 configurations at each temperature at $L=32$.
For the EA model at $T=0.5$, the mean absolute error is below $10^{-11}$.
For the triangular AFM model at $T=0.2$, the mean absolute error in the total $\ln Z$ is below $10^{-5}$, corresponding to less than $10^{-8}$ per spin.
All sampling results reported in this work were obtained in IEEE 754 double precision, which is therefore sufficient for the parameter ranges used in the experiments.
Arbitrary-precision arithmetic becomes necessary only at sufficiently low temperatures or for larger systems.

\begin{figure*}[!ht]
    \centering
    \includegraphics[width=\linewidth]{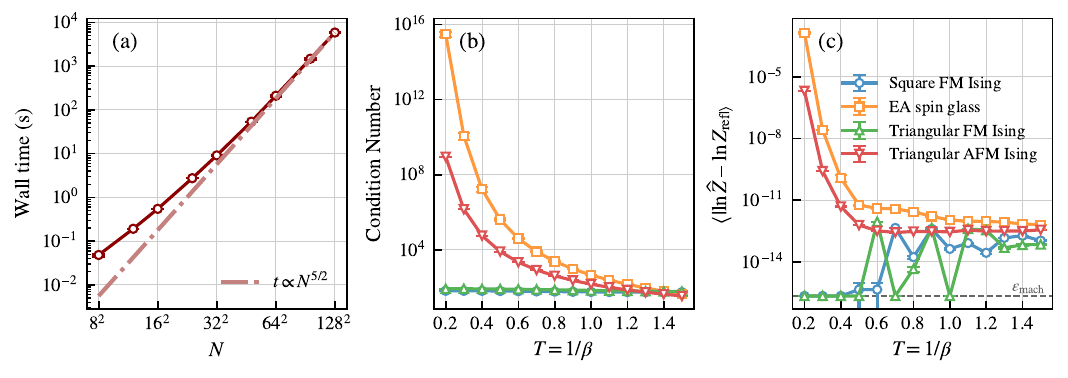}
    \caption{(a) Wall-clock time versus the number of spins $N$ for sampling one full spin configuration of the square FM Ising model.
    (b) Condition number of the matrix $A=I - B D^{-}$ as a function of temperature for sampling the first spin on an $L=32$ lattice.
    While the square and triangular FM Ising models maintain condition numbers of $\mathcal{O}(1)$--$\mathcal{O}(10^{2})$ across all temperature regimes, frustrated systems such as the EA spin glass and the triangular AFM Ising model exhibit an exponential divergence as $T \to 0$.
    (c) Mean absolute error $\langle|\ln\widehat{Z} - \ln Z_{\mathrm{ref}}|\rangle$ of the sampling estimate $\ln\widehat{Z} = -\beta E(\s) - \ln \widehat{P}(\s)$, with $\ln Z_{\mathrm{ref}}$ obtained in 200-bit arbitrary precision (1000 samples per temperature, $L = 32$).
    The dashed line marks the double-precision machine epsilon.
    Unfrustrated models show small accumulated roundoff errors throughout, while frustrated systems lose accuracy rapidly as $T \to 0$.}
    \label{fig:appendix}
\end{figure*}

\section{Exact thermodynamic quantities from the Kac--Ward partition function}\label{appendix:c}
We derive closed-form expressions for the internal energy, the entropy, and the specific heat per spin by differentiating the Kac--Ward formula with respect to $\beta$.
Let $A=I-BD$ and introduce the diagonal matrices $D^{(1)}=\partial_\beta D$ and $D^{(2)}=\partial_\beta^2D$.
Their diagonal entries are
\begin{align}
    D^{(1)}_{ij,ij} &= J_{ij}\,\operatorname{sech}^2(\beta J_{ij}), \\
    D^{(2)}_{ij,ij} &= -2J_{ij}^2\,\tanh(\beta J_{ij})\,\operatorname{sech}^2(\beta J_{ij}).
\end{align}

The internal energy is $U=\langle E(\s)\rangle=-\partial_\beta\ln Z$.
Applying Jacobi's formula $\partial_\beta\ln\det A=\operatorname{Tr}(A^{-1}\partial_\beta A)$ to Eq.~\eqref{eq:kacward}, together with $\partial_\beta A=-BD^{(1)}$, we obtain
\begin{equation}\label{eq:internal-energy}
    U = -\sum_{\langle i,j\rangle}J_{ij}\,\tanh(\beta J_{ij}) + \frac{1}{2}\operatorname{Tr}\left(A^{-1}BD^{(1)}\right).
\end{equation}
Together with the free energy $F=-\beta^{-1}\ln Z$, the entropy per spin is obtained exactly as
\begin{equation}\label{eq:entropy}
    s = \frac{\beta\,(U-F)}{N} = \frac{\beta U+\ln Z}{N}.
\end{equation}

The specific heat per spin is
\begin{equation}
    c_v = \frac{\beta^2}{N}\frac{\partial^2 \ln Z}{\partial \beta^2}.
\end{equation}
Using the identity $\partial_\beta A^{-1}=-A^{-1}(\partial_\beta A)A^{-1}$, Eq.~\eqref{eq:kacward} gives
\begin{multline}\label{eq:specific-heat}
    c_v = \frac{\beta^2}{N}\sum_{\langle i,j\rangle}J_{ij}^2\,\operatorname{sech}^2(\beta J_{ij}) \\
    -\frac{\beta^2}{2N}\operatorname{Tr}\left[\left(A^{-1}BD^{(1)}\right)^2\right]
    -\frac{\beta^2}{2N}\operatorname{Tr}\left(A^{-1}BD^{(2)}\right).
\end{multline}

\bibliography{refs}

\end{document}